\documentclass[a4paper,12pt]{article}
\usepackage[utf8]{inputenc}
\usepackage{amsmath,amssymb,amsfonts,amsthm}
\newtheorem{theorem}{Theorem}[section]

\usepackage{graphicx}
\usepackage{wasysym}
\usepackage[none]{hyphenat}

\newcommand{\den}{\varepsilon}
\newcommand{\Rdm}{\mathbb{R}^2_{+}}

\newcommand{\Ldt}{{^{(3)}\Delta}}

\newcommand{\Ldta}{{^{(7)}}\Delta}

\newcommand{\ck}{\mathcal{L}}

\newcommand{\po}{d}

\newcommand{\mt}{\bar{m}_2}

\title{Bounds for the metric and shift vector on extreme Kerr under linear axially symmetric gravitational perturbations}
\author{Ivan Gentile de Austria\\
 Facultad de Ciencias Exactas y Naturales, \\
  Universidad Nacional de Cuyo, CONICET,\\
  Mendoza, Argentina.}

\begin{document}

\sloppy

\maketitle

\begin{abstract}
 We present integral estimates for the shift and induced metric in the maximal isothermal gauge, corresponding to linear axially 
 symmetric gravitational perturbations over the Minkowski and extreme Kerr metric. These bounds are a step towards finding pointwise bounds for the 
 perturbations. In particular the presented estimate for the shift includes the horizon for the extreme Kerr case.  
\end{abstract}

\section{Introduction}
\label{intro}

The stability of the Kerr black hole is one of the most fundamental open problems in classical general relativity, a lot of work has been devoted to solving it. 
The first approach was 
the modal stability proved by Whiting \cite{Whiting89} using the Teukolsky equation \cite{teukolsky1972rotating}. However,
modal stability is not enough to exclude that general linear perturbations grow
unbounded in time (see, for example, the discussion in \cite{wald:1056} and
\cite{dafermos2010black}). A second approach is the study of the wave equation over the Kerr metric, this was extensively done by Dafermos and Rodnianski, they proved the stability 
of the solution of the wave equation on the Kerr metric, see \cite{dafermos2010decay}, \cite{dafermos2014decay} and \cite{dafermos2014scattering}. But the non-modal 
stability for gravitational perturbations over the Kerr metric is still an open problem.

Within the Kerr family a relevant case is the extreme Kerr metric as it lies on the boundary between black holes and
naked singularities. It is expected that studying it would shed light on the
cosmic censorship conjecture. In \cite{Aretakis:2011gz} and \cite{aretakis2013nonlinear} Aretakis proved that there exists a type of instabilities for the derivatives 
transversal to the horizon for 
the scalar wave under extreme Kerr black hole and general extreme black holes respectively. 
 There are also stability results for the extreme Reissner-Nordström black hole as 
 the ones shown in \cite{Aretakis:2011ha}.

The study of stability for the Kerr metric under gravitational perturbations is very hard. A first approach is to deal with the linear problem and a second approximation 
is simplifying the 
problem imposing axial symmetry on the perturbation, that is what we have done in our previous 
works \cite{Dain:2014iba} and \cite{dain2015bounds}, and we continue now in the present article.  

For gravitational perturbations with axial symmetry a reduction of the manifold can be performed in which the coordinate system $(t,\,\rho,\,z)$ 
is specified by the maximal isothermal gauge (see \cite{Dain:2014iba}). 
The domain for the spatial coordinates $(\rho,z)$ is the half-plane $\Bbb R_{+}^{2}$ defined by $0<\rho<\infty$ and $-\infty<z<\infty$, 
with flat metric $\delta_{AB}$ and line element given by $\delta=d\rho^{2}+dz^{2}$. The axis of symmetry is given by $\rho=0$, and  the horizon is 
located at $r=0$, being
$r=\sqrt{\rho^{2}+z^{2}}$. In this case a four dimensional solution of the Einstein equations $(M,\, g_{\alpha\beta})$ corresponds to a three 
dimensional manifold $(N,\, h_{ab})$ being $N$ the quotient of $M$ with respect to the trajectories of the axial Killing vector $\eta^{\mu}$. 
Then $N$ is foliated by two dimensional slices of constant time coordinate, the slices have metric $q_{AB}=e^{2u}\delta_{AB}$ and are maximal, 
$\chi=0$. The degrees of freedom of the gravitational field are included in the norm of the Killing vector $\eta$ and the corresponding twist potential 
$\omega$. Due to the singular behavior at the axis $\rho=0$ it is better to work with the functions $\sigma$ and $q$, defined by $\eta=\rho^{2}e^{\sigma}$ and 
$u=\mbox{log}\,\rho+\sigma+q$.

As explained in \cite{Dain:2014iba}, the perturbations with axial symmetry can be characterized by two functions $\sigma_{1}$ and $\omega_{1}$, which 
represent the first order perturbations of the norm and the twist of the axial Killing vector. These functions are equally valid to describe the 
system of equations and perturbation as the functions $q_{1}$ and $\beta_{1}^{A}$, 
which represent the first order perturbation of the induced metric and shift vector respectively.
In \cite{Dain:2014iba} it was shown that in the presence of axial gravitational perturbations there are positive conserved energies, using these energies 
in \cite{dain2015bounds} we prove the existence of integral bounds for the first and second derivative of the perturbation. Said bounds do not control the 
perturbation on the horizon, which is where the equations diverge. In this article we will use both
results to give integral bounds for the first order of the shift vector and the induced metric. The results are obtained first
on Minkowski and then on extreme Kerr (Theorem (\ref{teorema4}) and Theorem (\ref{teorema5}) respectively).

There are several reasons for wanting to get bounds for $q_{1}$ and $\beta_{1}$. First, as a matter of completeness since they are equally valid 
functions to describe the system. Also, $q_{1}$ and $\beta_{1}$ are more usual variables for numerical purposes such as those developed in  
\cite{Dain:2009wi}. Our main reason can be described as follows. 
With the ambition to study the behavior of the perturbation at the horizon, we want to analyze the possibility of using in extreme Kerr 
the mechanism of \textit{time integration} developed in \cite{Dain:2012qw}, in which a pointwise bound is obtained for
the solution of the wave equation in the extreme Reissner-Nordström background including the horizon. For this, as it is shown in \cite{Dain:2012qw}, 
it is essential to give an integral estimate of the square of the temporal derivative of the wave, in our context this translates to the need to give an integral 
estimate of the squares of $\dot{\sigma}_{1}$ and $\dot{\omega}_{1}$. As we will see, to have a chance of controlling $\dot{\sigma}_{1}$ and $\dot{\omega}_{1}$ 
we need a bound for the shift vector, actually we need a bound for 
$\int_{\Bbb R_{+}^{2}}|\beta_{1}|^{2}\rho\, d\rho\, dz$. Furthermore, the bounds obtained in the present article for $q_{1}$ and $\beta_{1}$ differ from 
those obtained for $\sigma_{1}$ and $\omega_{1}$ in \cite{dain2015bounds}, 
which can shed light on the behavior of perturbations near the horizon.

The article is structured as follows. In sections \ref{ecuaciones_mink} and \ref{ecuaciones_kerr}, we present the  equations, definitions and identities that we will use 
in this article, 
for a complete exposition of them see 
\cite{Dain:2014iba} and \cite{dain2015bounds}. The background quantities are indicated with a subscript $0$ while those corresponding to the 
first order of the perturbation with a subscript $1$. Then, in section \ref{estimaciones} we present and prove the results for Minkowski and extreme Kerr 
in subsections \ref{estimaciones_mink} and \ref{estimacion_kerr} respectively.
Finally, in section \ref{conclusiones}, we comment and discuss the results.

\section{Minkowski Equations}
\label{ecuaciones_mink}

In this section we give a summary of the linearized equations and background quantities for axially symmetric perturbations 
of the Minkowski spacetime. We do this because we first obtain the results in Minkowski 
in section \ref{estimaciones_mink} as a matter of consistency with previous works and because it serves as a test case.

The background functions are
\begin{equation}
  \label{eq77}
  \omega_{0} = q_{0}= \sigma_{0}=0,\quad 
 u_{0}=\ln\rho,\quad  \eta_{0} = \rho^{2}, \quad \alpha_{0}=\rho.
\end{equation}
The evolution equations for $\sigma_{1}$ and $\omega_{1}$ are

\begin{equation}
 \label{eq79}
 - \dot p + \Ldt\sigma_{1} =0,\\
\end{equation}

\begin{equation} 
 \label{eq80}
 -\ddot{\omega}_{1} +\Ldt \omega_{1}=4\dfrac{\partial_{\rho}\omega_{1}}{\rho},
\end{equation}
where

\begin{equation}
 \label{eq80a}
 \Ldt=\Delta+\dfrac{\partial_{\rho}}{\rho},
\end{equation}

\begin{equation}
 \label{eq81}
 p=\dot{\sigma}_{1}-\dfrac{2\beta_{1}^{\rho}}{\rho}.  
\end{equation} 
The evolution equations for the metric and extrinsic curvature are

\begin{equation}
\label{eq82} 
 2\dot{u}_{1}=\partial_{A}\beta^{A}_{1}+2\dfrac{\beta^{\rho}_{1}}{\rho},\\
\end{equation}

\begin{equation} 
 \label{eq83}
 \dot{\chi}_{1AB} =2\partial_{(A}q_1\partial_{B)}\rho-\delta_{AB}\partial_{\rho}q_1,
\end{equation}
where the dot means derivative with respect to $t$.

The momentum and Hamiltonian constraint equations are

\begin{equation}
 \label{eq84}
 \partial^{A}\chi_{1AB} =-p\partial_{A}\rho, \\
\end{equation}

\begin{equation} 
 \label{eq85}
 \Delta q_{1}+^{(3)}\Delta\sigma_{1} =0.
\end{equation}
The gauge equations are

\begin{equation}
 \label{eq86}
 \alpha_{1}=0,
\end{equation}

\begin{equation}
 \label{eq86a}
 (\ck\beta_{1})^{AB}=\dfrac{2}{\rho} \chi_{1}^{AB}, 
\end{equation}
where 
\begin{equation}
 \label{ck}
 \ck\beta_{AB}=\partial_{A}\beta_{B}+\partial_{B}\beta_{A}-\delta_{AB}\partial_{C}\beta^{C}.
\end{equation}
For the energy densities and mases we found 

\begin{equation}
  \label{eq87}
  \den_0=\den_1=0,\quad
   m_0=m_1=0,
\end{equation}

\begin{equation}
\label{eq88}
 \den_{2}=\left(2p^{2}+2\dfrac{\dot
     {\omega}_{1}^{2}}{\rho^{4}}+2\arrowvert\partial\sigma_{1}\arrowvert^{2}+2\dfrac{\arrowvert\partial 
   \omega_{1}\arrowvert^{2}}{\rho^{4}}+4\dfrac
 {\chi_{1}^{AB}\chi_{1AB}}{\rho^{2}}\right) \rho, 
\end{equation}

\begin{equation}
 \label{masa_mink_seg_ord}
 m_{2}=\int_{\Bbb R_{+}^{2}}\den_{2}d\rho dz,
\end{equation}
where $m_{0}$, $m_{1}$ and $m_{2}$ are the zero, first and second order ADM mass, and their corresponding densities $\den_{0}$, $\den_{1}$ and $\den_{2}$ respectively.

\section{Extreme Kerr Equations}
\label{ecuaciones_kerr}

Now we give a summary of the linearized equations and background quantities for the extreme Kerr metric. These equations will be used in section \ref{estimacion_kerr}.

For the background quantities
\begin{equation}
   \label{eq128}
   \alpha_0=\rho,
 \end{equation}
and the explicit expressions for $q_{0},\,\sigma_{0},\,\omega_{0}$ are complicated and can be seen in the appendix of \cite{Dain:2014iba}. 
These functions depend on two parameters, the mass 
$m_{0}$ and the angular momentum per unit of mass $a$ of the black hole. The relevant properties are

\begin{equation}
  \label{eq129}
  \frac{|\partial \omega_0|^{2}}{\eta^2_0} \leq \frac{C}{r^{2}}, 
\end{equation}

\begin{equation}
 \label{eq129a}
 |\partial \sigma_0|^2 \leq \frac{C}{r^{2}},
\end{equation}

\begin{equation}
 \label{eq129a1}
 |\partial q_{1}|^{2}\leq C
\end{equation}

\begin{equation}
 \label{eq:combi}
  e^{-2 (\sigma_0+q_0)} \frac{|\partial \omega_0|^{2}}{\eta^2_0}  \leq C, 
\end{equation}

\begin{equation}
\label{eq:combi2}
e^{-2(\sigma_{0}+q_{0})}|\partial\sigma_{0}|^{2}\leq C,
\end{equation}
where the positive constant $C$ depends only on $m_{0}$. 
The evolution equations for $\sigma_{1}$ and $\omega_{1}$ are

\begin{align}
 \label{eq131}
 -\frac{e^{2u_{0}}}{\rho^{2}}\dot{p}+^{(3)}\Delta\sigma_{1} &=\dfrac{2}{\eta^2_0}\left(\sigma_{1}|\partial\omega_{0}|^{2}-\partial_{A}
 \omega_{1}\partial^{A}\omega_{0} \right),\\
 \label{eq131a}
 -\frac{e^{2u_{0}}}{\rho^{2}}\dot{\po}+^{(3)}\Delta\omega_{1} &=4\dfrac{\partial_{\rho}\omega_{1}}{\rho}+
 2\partial_{A}\omega_{1}\partial^{A}\sigma_{0}+2\partial_{A}\omega_{0}\partial^{A}\sigma_{1},
\end{align}
where we define

\begin{align}
 \label{eq132}
 p &=\dot{\sigma}_{1}-2\dfrac{\beta_{1}^{\rho}}{\rho}-\beta_{1}^{A}\partial_{A}\sigma_{0},\\
 \label{eq132a}
 \po &=\dot{\omega}_{1}-\beta_{1}^{A}\partial_{A}\omega_{0}.
\end{align}

Unlike in Minkowski in this case no significant simplifications occur in the evolution equations for the metric and the second fundamental 
form with respect to the general ones shown in \cite{Dain:2014iba}, in addition we will not use these equations, therefore we do not present them here.

The momentum and Hamiltonian constraint equations are given by
\begin{align}
 \label{eq133}
 \partial^{B}\chi_{1AB} &=-\dfrac{e^{2u_{0}}}{2\rho}\left(p\left(2\dfrac{\partial_{A}\rho}{\rho}
 +\partial_{A}\sigma_{0}\right)
+\dfrac{\partial_{A}\omega_{0}}{\eta_{0}^{2}}\po\right),\\
 \label{eq133a}
 ^{(3)}\Delta\sigma_{1}+\Delta q_{1} &=- \frac{\den_1}{4\rho },
\end{align}
where $\den_1$ is given by (\ref{eq136a}).

The gauge equations are

\begin{equation}
 \label{eq86}
 \alpha_{1}=0
\end{equation}

\begin{equation}
 \label{eq135}
 \left(\ck\beta_{1}\right)^{AB}=2e^{-2u_{0}}\rho \chi_{1}^{AB},
\end{equation}
while for the energy densities we get

\begin{equation}
 \label{eq136}
 \den_{0}=\left(|\partial\sigma_{0}|^{2}+\dfrac{|\partial\omega_{0}|^{2}}{\eta_{0}^{2}}\right)\rho,
\end{equation}

\begin{equation}
 \label{eq136a}
 \den_1=\left( 2\partial_{A}\sigma_{0}\partial^{A}\sigma_{1} +\frac{2\partial_{A}\omega_{0}\partial^{A}\omega_{1}}{\eta_{0}^{2}}-\frac{2\sigma_{1}|\partial\omega_{0}|^2}{\eta_{0}^{2}}\right)\rho,
\end{equation}

\begin{multline}
 \label{eq136b} 
 \den_2=\left(2\dfrac{e^{2u_{0}}}{\rho^{2}}p^{2}+ 2\dfrac{e^{2u_{0}}}{\rho^{2}\eta_{0}^{2}}\po^{2} 
 +4e^{-2u_{0}} \chi^{AB}_{1}\chi_{1AB} +\right. \\
 +\left(\partial\sigma_{1}+\omega_{1}\eta^{-2}_{0}\partial\omega_{0}\right)^{2} +\left(\partial\left(\omega_{1}\eta_{0}^{-1}\right)-\eta_{0}^{-1}\sigma_{1}\partial\omega_{0}\right)^{2}+\\
 \left. +\left(\eta_{0}^{-1}\sigma_{1}\partial\omega_{0}-\omega_{1}\eta_{0}^{-2}\partial\eta_{0}\right)^{2}\right)\rho, 
\end{multline}
with the corresponding masses

\begin{equation}
 \label{eq136c}
 m_{0}=\int_{\Bbb R_{+}^{2}}\den_{0}d\rho dz,
\end{equation}

\begin{equation}
 \label{eq136c1}
 m_{1}=\int_{\Bbb R_{+}^{2}}\den_{1}d\rho dz,
\end{equation}

\begin{equation}
 \label{eq136c2}
 m_{2}=\int_{\Bbb R_{+}^{2}}\den_{2}d\rho dz.
\end{equation}
As shown in \cite{Dain:2014iba}, due to the background being stationary, taking the time derivative in all equations (both in Minkowski and in extreme Kerr) we get again the same set of equations
but now for the time derivative of the perturbation quantities ($\dot{\sigma}_{1},\; 
\dot{\omega}_{1},\;\dot{q}_{1},\;\dot{\chi}_{1AB},$ etc.), so we have the corresponding conserved masses ($\bar{m}_{0},\;\bar{m}_{1},\;
\bar{m_{2}}$) which are obtained directly by replacing in $m_{0},\; m_{1},\; m_{2}$ the quantities for their time derivative.     

In this article as in \cite{Dain:2014iba} and \cite{dain2015bounds} the perturbation satisfies the usual decay conditions for asymptotically flat spacetimes 
with a cylindrical end 
(see \cite{Dain:2014iba}).

We will also make use of the following theorems, proved in \cite{dain2015bounds}.
\begin{theorem}
 \label{theorem3}
Linear gravitational perturbations with axial symmetry over extreme Kerr satisfy the following inequalities
\begin{align}
 \label{cota1}
 & \int_{ \Rdm}\left(\dfrac{1}{2}\eta_{0}^{2}|\partial\bar{\omega}_{1}|^{2}+|\partial\eta_{0}|^{2}\bar{\omega}_{1}^{2}+  |\partial\sigma_{1}|^{2}+\dfrac{\sigma_{1}^{2}}{r^{2}}\right)\rho d\rho dz\leq  C m_{2} , \\
 \label{cota2}
 & \int_{\Rdm}\left( \left(\Ldt\sigma_{1}\right)^{2}+\eta_{0}^{2}\left(\Ldta\bar{\omega}_{1}\right)^{2} \right)e^{-2(q_{0}+\sigma_{0})} \rho d\rho dz\leq  C \left(\mt+m_{2} \right),
 \end{align}
where $C$ is a constant that depends only on the extreme Kerr mass $m_{0}$, and $\bar{\omega}_{1}$ is defined as

\begin{equation}
 \label{eq_omega_barra}
 \bar \omega_1 =\frac{\omega_1}{\eta_0^{2}}.
\end{equation}

\end{theorem}  
We will also use the generalized Cauchy inequality
\begin{equation}
 \label{eq:desigualdaddecauchy1}
 a_{1}^{2}+a_{2}^{2}+...\,+a_{n}^{2}\geq\dfrac{1}{n}\left(a_{1}+a_{2}+...\,+a_{n}\right)^{2}.
\end{equation}

\section{Estimates for $q$ and $\beta^{A}$} 
\label{estimaciones}
In this section we present and prove our results. First for Minkowski as background and then for extreme Kerr. 
We first obtain the results in Minkowski as a matter of consistency with previous works. In addition, as we shall see, the result in Minkowski is broader than in the 
case of extreme Kerr.  

As we mentioned in the introduction, the motivation to obtain bounds for the shift follows from the necessity to give bounds for the time derivative of the norm of the axial Killing 
vector and the corresponding twist potential. This can be seen from equations (\ref{eq132}) and (\ref{eq132a}) for extreme Kerr (or equation (\ref{eq81}) for Minkowski), from which 
follows the inequalities 
\begin{equation}
 \label{int_1}
 \int_{\Bbb R_{+}^{2}}|\dot{\sigma}_{1}|^{2}e^{2(\sigma_{0}+q_{0})}\rho d\rho dz\leq m_{2}+
 \int_{\Bbb R_{+}^{2}}\left(\dfrac{|\beta_{1}^{\rho}|^{2}}{\rho^{2}}+|\beta_{1}^{A}|^{2}|\partial\sigma_{0}|^{2}\right)e^{2(\sigma_{0}+q_{0})}\rho d\rho dz
\end{equation}
and
\begin{equation}
 \label{int_2}
 \int_{\Bbb R_{+}^{2}}|\dot{\omega}_{1}|^{2}\dfrac{e^{2(\sigma_{0}+q_{0})}}{\eta_{0}^{2}}\rho d\rho dz\leq m_{2}+
 \int_{\Bbb R_{+}^{2}}|\beta_{1}^{A}|^{2}\dfrac{e^{2(\sigma_{0}+q_{0})}}{\eta_{0}^{2}}|\partial\omega_{0}|^{2}\rho d\rho dz.
\end{equation}
Integrand terms different from $\beta_{1}$ are bounded on $\Bbb R_{+}^{2}$ so they can be taken off the integral, then we need to control 
$\int_{\Bbb R_{+}^{2}}|\beta_{1}|^{2}\rho d\rho dz$.

\subsection{Minkowski}
\label{estimaciones_mink}

\begin{theorem}
 \label{teorema4}
For linear gravitational perturbations with axial symmetry on Minkowski, the following inequalities hold 
\begin{equation}
 \label{cota5}
 \int_{\Bbb R_{+}^{2}}\dfrac{|\partial q_{1}|^{2}}{\rho}d\rho dz\leq C\bar{m}_{2},
\end{equation}

\begin{equation}
 \label{cota6}
 \int_{\Bbb R_{+}^{2}}|\Delta q_{1}|^{2} \rho d\rho dz\leq C\bar{m}_{2},  
\end{equation}

\begin{equation}
 \label{cota7}
 \int_{\Bbb R_{+}^{2}}|\Delta\beta_{1}^{B}|^{2} \rho^{3} d\rho dz\leq Cm_{2},
\end{equation}
being $C>0$ a numerical constant. 

Moreover, (\ref{cota5}) and (\ref{cota6}) imply the following estimate for the second derivatives of $q_{1}$

\begin{equation}
 \label{cota8a}
\int_{\Bbb R_{+}^{2}}|\partial^{2}q_{1}|^{2}\rho d\rho dz\leq C\bar{m}_{2}, 
\end{equation}
where $|\partial^{2}q_{1}|^{2}=\partial^{A}\partial_{B}q_{1}\partial^{B}\partial_{A}q_{1}$.
\end{theorem}

\begin{proof}
We start with the estimate (\ref{cota5}). From the extrinsic curvature equation (\ref{eq83}), we have
\begin{equation}
 \label{2}
 \dot{\chi}_{1AB}\dot{\chi}_{1}^{AB}=2|\partial q_{1}|^{2},
\end{equation}
therefore multiplying by the factor $\dfrac{1}{\rho}$ and integrating on $\Bbb R_{+}^{2}$, it follows from (\ref{masa_mink_seg_ord}) that
\begin{equation}
 \int_{\Bbb R_{+}^{2}}\dfrac{|\partial q_{1}|^{2}}{\rho}d\rho dz\leq C\bar{m}_{2}.
\end{equation}
From (\ref{eq79}) and (\ref{eq85}) we have

\begin{equation}
 \label{3}
\Delta q_{1}=-\dot{p},
\end{equation}
squaring this equation, integrating on $\Bbb R_{+}^{2}$ and using the definition (\ref{masa_mink_seg_ord}) 
the bound (\ref{cota6}) is obtained. 

The estimate (\ref{cota7}) which contains the Laplacian of $\beta_{1}^{B}$ is obtained as follows. Taking the divergence of (\ref{eq86a}) we get 

\begin{align}
 \label{4}
\partial_{A}\left(\ck\beta_{1}\right)^{AB}&=\partial_{A}\left(\dfrac{2}{\rho}\chi_{1}^{AB}\right),\\ 
\label{5}
\Delta\beta_{1}^{B}&=\dfrac{2}{\rho}\left(\partial_{A}\chi_{1}^{AB}-\dfrac{\chi_{1}^{\rho B}}{\rho}\right),\\
\label{6}
\Delta\beta_{1}^{B}&=-\dfrac{2}{\rho}\left(p\partial^{B}\rho+\dfrac{\chi_{1}^{\rho B}}{\rho}\right),
\end{align}
where in the line (\ref{5}) we have used the operator $\ck$ given by (\ref{ck}), in order to get the equation (\ref{6}) we made use of the
constrain equation (\ref{eq84}). Multiplying (\ref{6}) by $\rho$ we can rewrite it as

\begin{equation}
 \label{7}
\Delta\beta_{1}^{B}\rho=-2\left(p\partial^{B}\rho+\dfrac{\chi_{1}^{\rho B}}{\rho}\right).
\end{equation}
Taking the square of this equation and using the Cauchy inequality we have

\begin{align}
 \label{8}
|\Delta\beta_{1}^{B}|^{2}\rho^{2}&=2\left(p\partial^{B}\rho+\dfrac{\chi_{1}^{\rho B}}{\rho}\right)^{2}\\ 
\label{8a}
&\leq4\left(p^{2}+\dfrac{\chi_{1}^{AB}\chi_{1AB}}{\rho^{2}}\right),
\end{align}
then, integrating over $\Bbb R_{+}^{2}$ we see from (\ref{masa_mink_seg_ord}) that the estimate (\ref{cota7}) holds.

Finally, the estimate (\ref{cota8a}) comes from the relationship

\begin{multline}
 \label{9}
|\Delta q_{1}|^{2}\rho=\partial^{A}\partial_{A}q_{1}\partial^{B}\partial_{B}q_{1}\rho=\partial^{A}(\partial_{A}q_{1}\partial^{B}\partial_{B}q_{1}\rho)
-\partial^{B}(\partial_{A}q_{1}\partial^{A}\partial_{B}q_{1}\rho)\\
+\partial^{A}(\partial_{A}q_{1}\partial_{B}q_{1}\partial^{B}\rho)
+ \partial^{A}\partial_{B}q_{1}\partial^{B}\partial_{A}q_{1}\rho
-2\partial^{A}\partial_{A}q_{1}\partial_{B}q_{1}\partial^{B}\rho
-\partial_{A}q_{1}\partial_{B}q_{1}\partial^{A}\partial^{B}\rho,
\end{multline}
integrating (\ref{9}), we see that the first three terms cancel out because of the divergence theorem and the decay conditions. 
The last term vanishes because $\partial^{A}\partial^{B}\rho=0$. Then we get

\begin{equation}
 \label{10}
\int_{\Bbb R_{+}^{2}}|\partial^{2}q_{1}|^{2}\rho d\rho dz=\int_{\Bbb R_{+}^{2}}|\Delta q_{1}|^{2}\rho d\rho dz+
2\int_{\Bbb R_{+}^{2}}\Delta q_{1}\partial_{\rho}q_{1}d\rho dz,
\end{equation}
squaring (\ref{10}) and using the Cauchy inequality, we obtain

\begin{align}
 \label{11}
\left(\int_{\Bbb R_{+}^{2}}|\partial^{2}q_{1}|^{2}\rho d\rho dz\right)^{2}&\leq2\left(\int_{\Bbb R_{+}^{2}}|\Delta q_{1}|^{2}\rho d\rho dz\right)^{2}+
8\left(\int_{\Bbb R_{+}^{2}}\Delta q_{1}\partial_{\rho}q_{1}d\rho dz\right)^{2},\\
\label{11a}
&=2\left(\int_{\Bbb R_{+}^{2}}|\Delta q_{1}|^{2}\rho d\rho dz\right)^{2}+8\left(\int_{\Bbb R_{+}^{2}}\Delta q_{1}\rho^{1/2}\dfrac{\partial_{\rho}q_{1}}{\rho^{1/2}}d\rho dz\right)^{2}\\
\label{11b}
&\leq2\left(\int_{\Bbb R_{+}^{2}}|\Delta q_{1}|^{2}\rho d\rho dz\right)^{2}+8\int_{\Bbb R_{+}^{2}}|\Delta q_{1}|^{2}\rho d\rho dz
\int_{\Bbb R_{+}^{2}}\dfrac{|\partial_{\rho}q_{1}|}{\rho}d\rho dz,\\
\label{11c}
&\leq10C\left(\bar{m}_{2}\right)^{2},
\end{align}
where in the line (\ref{11a}) we have multiplied and divided by $\rho^{1/2}$ the integrand of the second term. The line (\ref{11b}) follows from the Hölder inequality. 
Lastly the inequality (\ref{11c}) is obtained from the bounds (\ref{cota5}) and (\ref{cota6}). Then, taking the positive root of (\ref{11c}) we get (\ref{cota8a}).

\end{proof}

\subsection{Extreme Kerr}
\label{estimacion_kerr}

\begin{theorem}
\label{teorema5}
 For linear gravitational perturbations with axial symmetry over extreme Kerr, the following inequalities hold
\begin{equation}
 \label{cota9}
 \int_{\Bbb R_{+}^{2}}|\Delta q_{1}|^{2}e^{-2(\sigma_{0}+q_{0})}\rho d\rho dz\leq 
 C(\bar{m}_{2}+m_{2})
\end{equation}
and

\begin{equation}
 \label{cota10}
 \int_{\Bbb R_{+}^{2}}|\Delta \beta^{B}_{1}|^{2}e^{2(\sigma_{0}+q_{0})}\rho^{3} d\rho dz\leq Cm_{2},
\end{equation}
where $C>0$ is a numerical constant.
\end{theorem}

\begin{proof}
Taking the square of (\ref{eq133a}), from the Cauchy inequality we get
\begin{equation}
 \label{dem_kerr_1}
 |\Delta q_{1}|^{2}\leq C\left[|\partial\sigma_0|^{2}|\partial\sigma_1|^{2}+\dfrac{|\partial\omega_0|^{2}}{\eta_{0}^{4}}|\partial\omega_1|^{2}+
 \dfrac{|\partial\omega_0|^{4}}{\eta_{0}^{4}}|\sigma_1|^{2}+(\Ldt\sigma_{1})^{2}\right],
\end{equation}
where we have used the definition of $\den_1$ (\ref{eq136a}).
Multiplying this inequality by the factor $e^{-2(\sigma_{0}+q_{0})}$, it follows from the background inequalities (\ref{eq129})-(\ref{eq:combi2}) that 

\begin{equation}
 \label{dem_kerr_2}
 |\Delta q_{1}|^{2}e^{-2(\sigma_{0}+q_{0})}\leq C\left[|\partial\sigma_1|^{2}+\dfrac{|\partial\omega_1|^{2}}{\eta_{0}^{2}}+
 \dfrac{|\sigma_1|^{2}}{r^{2}}+(\Ldt\sigma_{1})^{2}e^{-2(\sigma_{0}+q_{0})}\right], 
\end{equation}
integrating this on $\Bbb R_{+}^{2}$, we get from the Theorem \ref{theorem3} the inequality (\ref{cota9}).

To obtain the bound (\ref{cota10}), we start by taking the derivative $\partial_{A}$ of the equation (\ref{eq135}) 

\begin{equation}
 \label{dem_kerr_3}
 \Delta\beta_{1}^{B}=-4e^{-(\sigma_{0}+q_{0})}\partial_{A}(\sigma_{0}+q_{0})\dfrac{\chi_{1}^{AB}}{\rho}-
 2p\left(2\dfrac{\partial^{B}\rho}{\rho}+\partial^{B}\sigma_{0}\right)+2\dfrac{\partial^{B}\omega_{0}}{\eta_{0}^{2}}d,
\end{equation}
where we made use of (\ref{eq133}) to rewrite $\partial_{A}\chi^{AB}_{1}$. Squaring (\ref{dem_kerr_3}) and making use of the Cauchy inequality, we see that

\begin{equation}
 \label{dem_kerr_4}
 |\Delta\beta_{1}|^{2}\leq C\left(e^{-4(\sigma_{0}+q_{0})}|\partial(\sigma_{0}+q_{0})|^{2}\dfrac{\chi_{1AB}\chi_{1}^{AB}}{\rho^{2}}
 +\dfrac{p^{2}}{\rho^{2}}+|\partial\sigma_{0}|^{2}p^{2}+\dfrac{|\partial\omega_{0}|^{2}}{\eta_{0}^{4}}d^{2}\right).
\end{equation}
Multiplying by the factor $e^{2(\sigma_{0}+q_{0})}\rho^{3}$ we obtain

\begin{align}
 \label{dem_kerr_5}
 |\Delta\beta_{1}|^{2}e^{2(\sigma_{0}+q_{0})}\rho^{3}\leq C\left(e^{-2(\sigma_{0}+q_{0})}|\partial(\sigma_{0}
 +q_{0})|^{2}\chi_{1AB}\chi_{1}^{AB}\rho\right.\\
 \left.+e^{2(\sigma_{0}+q_{0})}\rho p^{2}+e^{2(\sigma_{0}+q_{0})}\rho^{3}|\partial\sigma_{0}|^{2}p^{2}+
 e^{2(\sigma_{0}+q_{0})}\rho^{3}\dfrac{|\partial\omega_{0}|^{2}}{\eta_{0}^{4}}d^{2}\right),
\end{align}
from the background properties (\ref{eq129})-(\ref{eq:combi2}) we see that

\begin{equation}
 \label{dem_kerr_6}
 |\Delta\beta_{1}|^{2}\leq C\left(\dfrac{e^{-2(\sigma_{0}+q_{0})}}{\rho}\chi_{1AB}\chi_{1}^{AB}+
 2e^{2(\sigma_{0}+q_{0})}\rho p^{2}+\dfrac{e^{2(\sigma_{0}+q_{0})}\rho}{\eta_{0}^{2}}d^{2}\right)
\end{equation}
so using (\ref{eq136c2}) we get the sought inequality

\begin{equation}
 \int_{\Bbb R_{+}^{2}}|\Delta \beta^{B}_{1}|^{2}e^{2(\sigma_{0}+q_{0})}\rho^{3} d\rho dz\leq Cm_{2}.
\end{equation}
\end{proof}

\textbf{Remark.} From equation (\ref{eq135}) we can obtain the following estimate for the first derivatives of the shift vector 
\begin{equation}
 \label{cota 8}
 \int_{\Bbb R_{+}^{2}}\left[(\partial^{\rho}\beta^{\rho}-\partial^{z}\beta^{z})^{2}
 +(\partial^{\rho}\beta^{z}+\partial^{z}\beta^{\rho})^{2}\right]\rho d\rho dz\leq Cm_{2}. 
\end{equation}
This estimate is not strong enough to be included in the theorem because it does not give a bound for the first derivatives of $\beta_{1}$ 
in a separately form, but we will discuss it in the conclusions. The same estimate can be obtained for Minkowski.

\section{Conclusions and Comments}
\label{conclusiones}

In the present work we found integral estimates for the derivatives of the shift and the induced metric in the maximal isothermal gauge for linear gravitational perturbations 
with axial symmetry in Minkowski and extreme Kerr. 

The results are similar to the ones obtained in \cite{dain2015bounds}, and in a certain way follow from those. An important difference is 
that the estimate (\ref{cota10}) does not control the Laplacian $\beta_{1}$ on the axis of symmetry ($\rho=0$) because the factor 
$e^{2(\sigma_{0}+q_{0})}\rho^{3}$ vanishes there, but it controls it on the horizon ($r=0$) where this factor becomes infinite. This does not mean that we 
have control over $\beta_{1}$ in the horizon, for this we would need to be able to control the first derivatives of $\beta_{1}$.
In this respect we can say that the equation (\ref{cota 8}) 
gives information about the behavior of the first derivatives but not each of them separately ($|\partial\beta_{1}|^{2}$) as we would need to try to control
$\beta_{1}$. Obtaining an estimate for $\beta_{1}$ is crucial to obtain bounds for $\dot{\sigma}_{1}$ and $\dot{\omega}_{1}$.

Unlike in Minkowski, in extreme Kerr we did not obtain estimates for the first derivatives of $q_{1}$ because in this case the evolution equation for the second 
fundamental form is much more complicated.

Using $q_{1}$ and $\beta_{1}$ we lose control in the sense that we only obtain integral estimates for the Laplacian with some weight factor 
($e^{-2(\sigma_{0}+q_{0})}\rho$ for $\Delta q_{1}$ and $e^{2(\sigma_{0}+q_{0})}\rho^{3}$ for $\Delta\beta_{1}$), but we gain control over the horizon for the Laplacian of
$\beta_{1}$, which is a progress towards the inclusion of the horizon in the bounds, likewise the estimate (\ref{cota 8}) also includes the horizon. 
The difficulty with respect to (\ref{cota 8}) is that it does not allow the use of usual estimates of the Sobolev type (see \cite{Evans98}) due to the form in which 
the derivatives appear in the integrands and it can not be combined with (\ref{cota10}) to give estimates of all the second derivatives of $\beta_{1}$ 
as was obtained for Minkowski with respect to $q_{1}$.

\bibliographystyle{plain}
\bibliography{biblio}

\end{document}